\documentclass[11pt,oneside]{article}

\usepackage{palatino}
\usepackage[margin=1in]{geometry}
\usepackage{color,xcolor,soul}
\usepackage[abspage,user,savepos]{zref}
\usepackage[colorlinks, citecolor=blue]{hyperref}  
\usepackage{natbib}
\usepackage{fp}
\usepackage{framed}
\usepackage{cuted,balance}
\usepackage{setspace}

\usepackage{graphicx} 
\usepackage{epsfig} 
\usepackage{mathtools}
\usepackage{dsfont,amsmath} 
\usepackage{amssymb}  
\usepackage{algorithm}
\usepackage{algpseudocode}
\usepackage[capitalize]{cleveref}

\DeclareMathOperator*{\argmax}{argmax}
\newcommand{\E}{\mathrm{E}}

\newcommand{\BR}{\mathrm{br}}

\newcommand{\be}{\begin{equation}}
\newcommand{\ee}{\end{equation}}

\newcommand{\opi}{\overline{\pi}}

\newcommand{\oM}{\overline{M}}
\newcommand{\oS}{\overline{S}}
\newcommand{\os}{\overline{s}}
\newcommand{\wpi}{\widehat{\pi}}
\newcommand{\wPi}{\widehat{\Pi}}
\newcommand{\wU}{\widehat{U}}
\newcommand{\weps}{\widehat{\epsilon}}
\newcommand{\pq}{\pmb{q}}
\newcommand{\ppi}{\pmb{\pi}}

\newcommand{\uu}{\underline{u}}

\newcounter{lemma}
{\medskip}

{\medskip}

\newcounter{proposition}
\newenvironment{proposition}{\refstepcounter{proposition}\par\medskip
\noindent 
\textbf{Proposition \theproposition.} \em \rmfamily}
{\medskip}

\newcounter{theorem}
\newenvironment{theorem}{\refstepcounter{theorem}\par\medskip
\noindent 
\textbf{Theorem \thetheorem.} \em \rmfamily}
{\medskip}

\newcounter{corollary}
{\medskip}

\newcounter{definition}
\newenvironment{definition}[1]{\refstepcounter{definition}\par\medskip
\noindent 
\textbf{Definition \thedefinition~(#1)} \em \rmfamily}
{\medskip}

\newcounter{assumption}
\newenvironment{assumption}{\refstepcounter{assumption}\par\medskip
\noindent 
\textbf{Assumption \theassumption.} \em \rmfamily}
{\medskip}

\newcounter{remark}
\newenvironment{remark}{\refstepcounter{remark}\par\medskip
\noindent 
\textbf{Remark \theremark.} \em \rmfamily}
{\medskip}

\newcounter{example}

{\medskip}

\newenvironment{proof}{
   \indent \textit{Proof.} \rmfamily}{\hfill $\square$}
   
\begin{document}

\date{}

\title{\LARGE \bf
Finite-horizon Approximations and Episodic Equilibrium for Stochastic Games
}

\author{Muhammed O. Sayin
\thanks{M. O. Sayin is with Department of Electrical and Electronics Engineering, Bilkent University, Ankara T\"{u}rkiye.
        {\tt\small sayin@ee.bilkent.edu.tr}}%
}
\maketitle

\bigskip

\begin{center}
\textbf{Abstract}
\end{center}
This paper proposes a finite-horizon approximation scheme and introduces episodic equilibrium as a solution concept for stochastic games (SGs), where agents strategize based on the current state and episode stage. The paper also establishes an upper bound on the approximation error that decays with the episode length for both discounted and time-averaged utilities. This approach bridges the gap in the analysis of finite and infinite-horizon SGs, and provides a unifying framework to address time-averaged and discounted utilities. To show the effectiveness of the scheme, the paper presents episodic, decentralized (i.e., payoff-based), and model-free learning dynamics proven to reach (near) episodic equilibrium in broad classes of SGs, including zero-sum, identical-interest and specific general-sum SGs with switching controllers for both time-averaged and discounted utilities. 

\begin{spacing}{1.245}

\section{Introduction}
\label{Int}
Stochastic games (SGs), introduced by \citep{ref:Shapley53}, generalize Markov decision processes (MDPs) to \textit{non-cooperative} multi-agent environments and offer a powerful framework for modeling multi-agent interactions in reinforcement learning applications across robotics, automation, and control. Recent research has explored learning dynamics of \textit{self-interested} agents in SGs, particularly for zero-sum and identical-interest games \citep{ref:Leslie20,ref:Sayin22a,ref:Sayin21,ref:Sayin22b,ref:Baudin22}. These approaches often view SGs through the lens of ``stage games" played at each visited state, similar to the repeated play of games. However, the stage games in SGs are not necessarily stationary due to evolving agent strategies. This creates a significant hurdle in extending the classical repeated-game results (e.g., see the overview \citep{ref:Fudenberg09}) to SGs, especially for time-averaged SGs where stationary equilibrium might not even exist \citep{ref:Gillette57}.

This paper proposes \textit{finite-horizon approximations} to tackle this problem for both time-averaged and discounted utilities, and introduces the concept of \textit{episodic equilibrium}, where agents adapt their strategies based on the current state and the stage within fixed-length episodes. This approach captures periodic behavior relevant to real-world scenarios, similar to daily/weekly actions of humans or monthly/annual planning of businesses. The paper further quantifies the approximation level of the scheme with respect to the episode length for both discounted and time-averaged SGs (see Theorem \ref{lem:neari}). For example, the error bound decays geometrically with the episode length for the discounted cases.

To show the effectiveness of the approximation scheme, the paper presents episodic, decentralized (i.e., payoff-based), and model-free learning dynamics by generalizing the individual Q-learning dynamics, introduced by \citep{ref:Leslie05}, to SGs under finite-horizon approximations. The paper proves their almost-sure convergence to (near) episodic equilibrium in two-agent zero-sum SGs, identical-interest SGs, and a specific class of general-sum SGs with switching controllers provided that the rewards induce either zero-sum or potential games for both time-averaged and discounted cases (see Theorem \ref{thm:SG}). For example, one agent can have time-averaged utility while the other has discounted utility such that the underlying game is neither time-averaged nor discounted SG. 

The recent works \citep{ref:Leslie20,ref:Sayin22a,ref:Sayin21,ref:Sayin22b,ref:Baudin22} address the non-stationarity of the stage games through two-timescale learning dynamics where agents update their continuation payoff estimates (or its equivalent) at a slower timescale. Such timescale separation makes stage games quasi-stationary. Therefore, agents can track equilibrium in stage games whose payoffs are the continuation payoffs estimated through the classical learning dynamics. For example, fictitious play dynamics can reach equilibrium in zero-sum and identical-interest games \citep{ref:Fudenberg09}. Due to the tracking result, we can show the convergence of the continuation payoff estimates either based on the contraction property of the minimax-Bellman operator peculiar to zero-sum SGs or the quasi-monotonicity of the equilibrium values peculiar to identical-interest SGs. However, the convergence guarantees beyond these classes of games require the development of new technical tools specific to the structure of the underlying SG. 

The finite-horizon approximation mitigates the need for the development of new technical tools by bridging the gap between finite-horizon and infinite-horizon SGs so that we can address the infinite-horizon SG by addressing its finite-horizon version, for which the analysis can be relatively easier.
Since time-averaged or discounted utilities do not make much difference for finite-horizon SGs, the approximation scheme also provides a unifying framework to address time-averaged and discounted SGs. Indeed, the impact of future rewards decays geometrically in the discounted case. Therefore, agents can effectively focus on the initial rewards received for some time period (called \textit{effective horizon}) and ignore the rewards to be received afterwards \citep{ref:Sutton18}. However, this is not an effective approach for the time-averaged case where every reward has the same weight in their utilities. The finite-horizon approximation scheme can mitigate this issue with the periodic play of episodic strategies computed according to finite-horizon lengths. 

In a broader sense, the paper is also related to the literature on finite approximations for MDPs. Examples include the finite quantizations of continuum state spaces and/or continuum action spaces for both time-averaged and discounted MDPs (see the overview \citep{ref:SaldiBook}). Following the same trend, the paper focuses on finite-horizon approximations of infinite-horizon SGs for the non-cooperative multi-agent environments. 

The paper is organized as follows. Section \ref{sec:SGs} describes episodic equilibrium for SGs. Section \ref{sec:approximation} presents the finite-horizon approximation scheme. Sections \ref{sec:dynamics} and \ref{sec:convergence} present the learning dynamics and characterize their convergence properties, respectively. Section \ref{sec:conclusion} concludes the paper with some remarks. 

\section{Stochastic Games and Episodic Equilibrium}\label{sec:SGs}

We can characterize an $N$-agent SG by the tuple $\mathcal{S} = \langle S,(A^i,r^i)_{i=1}^N,p \rangle$, where $S$ is the \textit{finite} set of states, $A^{i}$ and $r^i(\cdot)$ are, resp., the \textit{finite} set of actions and the reward function for agent $i$, and $p(s_+\mid s,a)$ denotes the probability of transition from state $s\in S$ to $s_+\in S$ when agents take the joint action $a=(a^i)_{i=1}^N$ and $a\in A := \prod_{i=1}^NA^i$. At each stage $k=0,1,\ldots$, each agent $i$ simultaneously chooses an action $a_k^i\in A^i$ after observing the state $s_k$ and recalling the history of the game, i.e., previous states and actions $\{s_l,a_l\}_{l < k}$. Let agent $i$ randomize his/her actions independently. For example, agent $i$ can follow a behavioral strategy $\pi^i(\cdot)$ determining which action should be played with what probability contingent on the information $h_k := \{s_k,s_{k-1},a_{k-1},\ldots,s_0,a_0\}$. Denote the set of behavioral strategies by $\Pi^i$ for agent $i$. Then, given the strategy profile $\pi=(\pi^j)_{j=1}^N$ and the initial state $s\in S$, agent $i$'s utility function is given by
\be\label{eq:utility}
U^i(\pi;s) := \lim_{K\rightarrow\infty}\frac{1}{\sum_{k=0}^K (\gamma^i)^k}\E\left[\sum_{k=0}^{K} (\gamma^i)^k r^i(s_k,a_k)\mid s_0=s \right],
\ee
with the agent-specific discount factor $\gamma^i\in (0,1]$, where $(s_k,a_k)$ is the pair of the state and the joint action at stage $k$ and the expectation is taken with respect to the randomness induced by the transition kernel and the behavioral strategies. Due to the normalization, this utility function addresses both time-averaged $\gamma^i=1$ and discounted $\gamma^i\in (0,1)$ cases. For example, in the time-averaged case, i.e., when $\gamma^i=1$, we have
\be\label{eq:utilityTime}
U^i(\pi; s) := \lim_{K\rightarrow\infty}\frac{1}{K+1}\E\left[\sum_{k=0}^{K} (\gamma^i)^k r^i(s_k,a_k) \mid s_0 = s\right].
\ee
In the discounted case, i.e., when $\gamma^i\in (0,1)$, the limit and the expectation can interchange due to the monotone convergence theorem such that the utility function is given by
\be\label{eq:utilityDiscounted}
U^i(\pi; s) := (1-\gamma^i)\E\left[\sum_{k=0}^{\infty} (\gamma^i)^k r^i(s_k,a_k) \mid s_0 = s\right].
\ee

We say that a strategy profile $\pi_*=(\pi^i_*)_{i=1}^N$ is \textit{$\epsilon$-Nash equilibrium} for the SG $\mathcal{S}$ with initial state $s\in S$ provided that
\be\label{eq:NE}
U^i(\pi^i_*,\pi^{-i}_*; s) \geq U^i(\pi^i,\pi_*^{-i};s) - \epsilon^i
\ee
for all $\pi^i\in\Pi^i$ and $i=1,\ldots,N$, where $\epsilon = (\epsilon^i)_{i=1}^N$.

Agents might look for periodic behavior inherent to the underlying environment, similar to daily/weekly behaviors of humans or monthly/annual planning of businesses. For example, consider $M$-length episodes and call the stages within an episode by substages. Then, agents can have Markov strategies in which the probability of an action gets taken depends not only on the current state but also on the current substage within the episode, i.e., $\opi^i:\oS \rightarrow \Delta^i$, where $\oS:=S\times\oM$, and $\oM := \{0,1,\ldots,M-1\}$.\footnote{Let $\Delta^i$ denote the probability simplex over the finite set $A^i$} We call such $\opi^i$ by \textit{episodic strategy}. 

\begin{definition}[Episodic Equilibrium]
We say that strategy profile $\opi=(\opi^i)_{i=1}^n$ is $\epsilon$-episodic equilibrium for the SG $\mathcal{S}$ provided that $\opi$ satisfies \eqref{eq:NE}, i.e., $\opi$ is $\epsilon$-Nash equilibrium, and $\opi^i:\oS \rightarrow \Delta^i$ is episodic strategy for each $i$.
\end{definition}

For $M=1$, episodic strategies reduce to Markov stationary strategies and episodic equilibria reduce to stationary equilibria, introduced by \citep{ref:Shapley53}. The following proposition addresses the existence of episodic equilibria in discounted SGs.

\begin{proposition}
Episodic equilibrium exists for any finite-length episodes in every (finite) discounted SGs.
\end{proposition}

\begin{proof}
The proof follows from the observation that if we extend the underlying state space $S$ to $\oS = S\times \oM$ by incorporating the substage index $m\in \oM$ to the underlying state $s\in S$, then episodic strategies would be contingent only on the current extended state $(s,m)\in \oS$ while the reward and transition kernels would still depend only on the current state and actions. Therefore, the episodic strategies in the underlying SG would be Markov stationary strategies in the SG over the extended state space and Markov stationary equilibrium always exists in discounted SGs \citep{ref:Fink64}.
\end{proof}

\section{Finite-horizon Approximation}\label{sec:approximation}

This section presents the finite-horizon approximation scheme establishing a bridge in the analyses of finite-horizon and infinite-horizon SGs. Given an SG $\mathcal{S} =  \langle S,(A^i,r^i)_{i=1}^N,p \rangle$, we can characterize its finite-horizon version by the tuple $\widehat{\mathcal{S}} = \langle M, S,(A^i,r^i)_{i=1}^N,p \rangle$, where $M$ is the finite horizon length. The finite state and action sets, reward and probability transition functions are as described in Section \ref{sec:SGs}. Let $\wpi^i(\cdot)$ denote agent $i$'s behavioral strategy contingent on the history of the play. Correspondingly, $\wPi^i$ denotes the set of behavioral strategies for agent $i$. Similar to \eqref{eq:utility}, given the strategy profile $\wpi = (\wpi^i\in \wPi^i)_{i=1}^N$ and the initial state $s\in S$, agent $i$'s utility is given by
\be\label{eq:utilityfinite}
\wU^i(\wpi;s) := \E\left[\sum_{k=0}^{M-1} (\gamma^i)^k r^i(s_k,a_k) \mid s_0 = s \right],
\ee
where the expectation is taken with respect to the randomness on the pair $(s_k,a_k)$ induced by $\wpi$ and $p(\cdot |\cdot)$.
We do not have normalization in the utility $\wU^i$, different from \eqref{eq:utility}, since the finite sum in \eqref{eq:utilityfinite} is bounded.

Episodic strategies for infinite-horizon SGs reduce to Markov strategies in finite-horizon SGs. Therefore, agents can \textit{learn/seek/compute} episodic equilibrium under finite-horizon approximation scheme by viewing the underlying infinite-horizon SG $\mathcal{S}=\langle S,(A^i,r^i)_{i=1}^2,p\rangle$ as the \textit{repeated play} of its finite-horizon version $\widehat{\mathcal{S}}=\langle M,S,(A^i,r^i)_{i=1}^2,p\rangle$ by partitioning the infinite horizon into $M$-length episodes. Across the repeated play of these finite-horizon versions, their initial states get determined based on the trajectory of the states within the underlying infinite-horizon SG.

The following theorem shows that Markov strategies attaining $\weps$-Nash equilibrium in $\widehat{\mathcal{S}} = \langle M, S,(A^i,r^i)_{i=1}^N,p \rangle$ can attain near episodic equilibrium in its infinite-horizon version $\mathcal{S}$. 

\begin{theorem}\label{lem:neari}
Consider an SG characterized by $\mathcal{S} = \langle S,(A^i,r^i)_{i=1}^n,p\rangle$. Let $\opi=(\opi^i)_{i=1}^N$ be an episodic strategy profile satisfying 
\be\label{eq:nearfinite}
\wU^i(\opi^i,\opi^{-i};s) \geq \wU^i(\wpi^i,\opi^{-i};s) - \weps^i
\ee
for all $\wpi^i\in\wPi^i$ and $i=1,\ldots,N$, for some error $\weps^i\geq 0$ and all $s\in S$ in the finite-horizon version $\widehat{\mathcal{S}} = \langle M,S,(A^i,r^i)_{i=1}^n,p\rangle$. Then, $\opi$ is $\epsilon$-episodic equilibrium for $\mathcal{S}$ with any initial state $s$ such that
\be
U^i(\opi^i,\opi^{-i};s) \geq U^i(\pi^i,\opi^{-i};s) - \epsilon^i
\ee
for all $\pi^i \in \Pi^i$ and $i=1,\ldots,N$.
For each $i$, the approximation error $\epsilon^i\geq 0$ is given by
\be\label{eq:bound}
\epsilon^i = \left\{\begin{array}{ll} 
(\delta^i+\weps^i)\cdot \cfrac{1}{M} & \mbox{if }\gamma^i=1\\ 
((\gamma^i)^M \cdot \delta^i + \weps^i)\cdot\cfrac{1-\gamma^i}{1-(\gamma^i)^M}& \mbox{if } \gamma^i\in [0,1)
\end{array}\right.,
\ee 
where
\begin{align}
&\delta^i := \max_{s\in S}\left\{\wU^{i}(\opi^i,\opi^{-i};s)\right\} - \min_{s\in S}\left\{\wU^{i}(\opi^i,\opi^{-i};s)\right\}.\label{eq:deltai}
\end{align}
\end{theorem}

\begin{proof}
The proof follows from induction. To this end, we first introduce
\be
U_{\underline{m}:\overline{m}}^{i}(\opi;s) := \E\left[\sum_{k=\underline{m}}^{\overline{m}}(\gamma^i)^{k-\underline{m}}r^i(s_k,a_k)\mid s_{\underline{m}} = s\right]\label{eq:uu}
\ee
for notational simplicity. Note that we have $U^i_{\ell M:(\ell+1)M-1}\equiv \wU^i$ for any $\ell =0,1,\ldots$ by \eqref{eq:utilityfinite}. Partition the time interval from $0$ to $LM-1$ for some $L=1,2,\ldots$ into $L$-many $M$-length intervals, e.g., from $\ell M$ to $(\ell+1)M-1$ for each $\ell = 0,1,\ldots$. Then we can decompose the expected value of the discounted sum of the rewards for any given strategy profile $(\pi^i,\opi^{-i})$ as
\begin{align}
&U_{0:LM-1}^i(\pi^i,\opi^{-i};s_0) = \sum_{\ell=0}^{L-1}(\gamma^i)^{\ell M}\E\left[\wU^i(\pi^i,\opi^{-i};s_{\ell M})\right],\label{eq:decomposition}
\end{align}
where the expectations on the right-hand side are taken with respect to the randomness on the state $s_{\ell M}$. Note that $\opi^{-i}$ are $M$-episodic strategies though $\pi^i$ can be any strategy from $\Pi$. 

Let $\opi_*^i$ be the best response to $\opi^{-i}$ in $\wU_{0:M-1}^i(\cdot; s_0)$ for any initial state $s_0$. Then, the decomposition \eqref{eq:decomposition} yields that
\begin{align}
U_{0:LM-1}^i(\pi^i,\opi^{-i};s_0) &\leq \sum_{\ell=0}^{L-1}(\gamma^i)^{\ell M} \E\left[\wU^i(\opi_*^i,\opi^{-i};s_{\ell M})\right]\\
&\leq \wU^i(\opi_*^i,\opi^{-i};s_0) + \max_{s\in S} \left\{\wU^i(\opi^i_*,\opi^{-i};s)\right\}\times \sum_{\ell=1}^{L-1}(\gamma^i)^{\ell M}
\end{align}
for any $\pi^i\in \Pi^i$. Then, by \eqref{eq:nearfinite}, we obtain \begin{align}
U_{0:LM-1}^i(\pi^i,\opi^{-i};s_0) \leq &\;\wU^i(\opi^i,\opi^{-i};s_0) + \max_{s\in S} \left\{\wU^i(\opi^i,\opi^{-i};s)\right\}\times \sum_{\ell=1}^{L-1}(\gamma^i)^{\ell M} + \weps^i\sum_{\ell=0}^{L-1}(\gamma^i)^{M\ell}\label{eq:upp}
\end{align}
for any $\pi^i\in \Pi^i$.
On the other hand, we have
\begin{align}
\max_{\pi^i\in \Pi^i} \left\{U_{0:LM-1}^i(\pi^i,\opi^{-i};s_0)\right\}&\geq U_{0:LM-1}^i(\opi^i,\opi^{-i};s_0) \\
&\geq \wU^i(\opi^i,\opi^{-i};s_0) + \min_{s\in S}\left\{\wU^i(\opi^i,\opi^{-i};s)\right\}\times \sum_{\ell=1}^{L-1}(\gamma^i)^{\ell M},\label{eq:low}
\end{align}
where the second inequality follows from \eqref{eq:decomposition}.
The upper and lower bounds, resp., \eqref{eq:upp} and \eqref{eq:low}, yield that
\begin{align}
&\max_{\pi^i\in \Pi^i} \left\{U_{0:LM-1}^i(\pi^i,\opi^{-i};s_0)\right\} - U_{0:LM-1}^i(\opi^i,\opi^{-i};s_0)\leq \delta^i \sum_{\ell=1}^{L-1}(\gamma^i)^{\ell M} + \weps^i\sum_{\ell=0}^{L-1}(\gamma^i)^{\ell M},\label{eq:mainbound}
\end{align}
where $\delta^i$ is as described in \eqref{eq:deltai}. Therefore, by \eqref{eq:utility} and \eqref{eq:uu}, we obtain the upper bounds \eqref{eq:bound} for both time-averaged and discounted utilities.
\end{proof} 

\begin{remark}
The term $\delta^i$ can be inherently small depending on the underlying game. For example, in zero-sum skew-symmetric games, i.e., $r^i(a^i,a^j) =-r^i(a^j,a^i)$ and $r^i(a^i,a^j) =-r^j(a^j,a^i)$ for each $i\neq j$, the equilibrium value is zero. Hence, $\delta^i=0$ for any $M$ if the rewards have the zero-sum skew-symmetric structure. 
\end{remark}

\begin{remark}
The term $\delta^i\rightarrow 0$ as $M\rightarrow \infty$ for irreducible SGs in the time-averaged cases since the value of the game does not depend on the initial state. In the discounted cases, the impact of $\delta^i$ on the upper bound decays to zero geometrically as $M\rightarrow\infty$ since $\delta^i \leq 2\max_{(s,a)}\{|r^i(s,a)|\}/(1-\gamma^i)$ for any $M$.
\end{remark}

The following section presents new dynamics for learning episodic equilibrium in SGs based on the finite-horizon approximation scheme to illustrate/exemplify its effectiveness.

\section{Episodic Individual Q-learning}\label{sec:dynamics}

This section generalizes the individual Q-learning dynamics \citep{ref:Leslie05} to finite-horizon SGs. In the episodic individual Q-learning dynamics, each agent $i$ assumes that the opponent $j\neq i$ plays according to some Markov strategy $\wpi^{j}:\oS\rightarrow\Delta^j$ across the repeated play of the finite-horizon SG $\widehat{\mathcal{S}}$ (i.e., the opponent plays according to some $M$-episodic strategy in the infinite-horizon SG $\mathcal{S}$), similar to the stationary-opponent-play assumptions adopted in the widely studied fictitious play dynamics and its variants, e.g., see \citep{ref:Fudenberg09}. This assumption yields that agent $i$ faces a finite-horizon MDP. Given $\widehat{\pi}^j$, agent $i$'s local $Q$-function $q_{\widehat{\pi}^j}^i:\oS\times A^i\rightarrow\mathbb{R}$ is given by\footnote{Henceforth, given any functional $f:X\rightarrow \mathbb{R}$, we let $f(\mu) := \E_{x\sim \mu}[f(x)]$ for any probability distribution $\mu$ over $X$ for the ease of exposition.} 
\begin{align}
&q^i_{\widehat{\pi}^j}(\os,a^i) := r^i(s,a^i,\wpi^{j}) + \gamma^i \sum_{s_+\in S}p(s_+\mid s,a^i,\wpi^j) \cdot v^i_{\widehat{\pi}^j}(\os_+)\label{eq:q}
\end{align}
for all $\os=(s,m)$ and $a^i$, where $\os_+=(s_+,m+1)$, and the value function $v^i_{\widehat{\pi}^j}:\oS\rightarrow\mathbb{R}$ is given by
\be\label{eq:v}
v^i_{\widehat{\pi}^j}(\os) = \mathbb{I}_{\{m<M\}}\cdot \max_{a^i\in A^i} \{q^i_{\widehat{\pi}^j}(\os,a^i)\}\quad\forall \os=(s,m),
\ee
where $\mathbb{I}_{\{\cdot\}}\in \{0,1\}$ is the indicator function. 

Agent $i$ does not know the opponent's strategy $\widehat{\pi}^j$ to compute $q_{\widehat{\pi}^j}^i(\cdot)$ by \eqref{eq:q}. Let $q_{k}^i:\oS\times A^i\rightarrow\mathbb{R}$ and $v_{k}^i:\oS\rightarrow\mathbb{R}$ denote, resp., agent $i$'s local Q-function and value function estimates at stage $k=0,1,\ldots$ based on the observations made. Given $q^i_k(\cdot)$, agent $i$ can take action $a_k^i\sim \BR^i_k:=\BR^i(q_k^i(\os_k,\cdot))$ at the current extended state $\os_k\in \oS$ according to the smoothed best response $\BR^i:\mathbb{R}^{A^i}\rightarrow \Delta^i$ defined by
\begin{equation}\label{eq:smoothed}
\BR^i(q^i) := \argmax_{\mu^i\in\Delta^i}\left\{ q^i(\mu^i) + \tau H(\mu^i)\right\},
\end{equation} 
where $H(\mu^i) := \E_{a^i\sim \mu^i}[-\log(\mu^i(a^i))]$ for any $\mu^i\in\Delta^i$ is the entropy regularization and $\tau>0$ controls the level of exploration in the response \citep{ref:Hofbauer05}. There exists unique smoothed best response due to the strong concavity of the entropy regularization. 

By \eqref{eq:v} and \eqref{eq:smoothed}, the value of the current $\os_k$ is given by
\be\label{eq:vupdate}
v^{i}_k(\os_k) = \mathbb{I}_{\{m_k<M\}} \cdot q_k^{i}(\os_k,\BR^{i}_k)
\ee
and $v_{k}^i(\os) = v_{k-1}^i(\os)$ for all $\os\neq \os_k$. Based on \eqref{eq:q} and $v_k^i(\cdot)$, agent $i$ receives the payoff 
\be\label{eq:sample}
\widehat{q}_k^i := r_k^i + \gamma^i\cdot v_k^i(\os_{k+1})\in\mathbb{R}
\ee
in the current stage game, where $r_k^i=r^i(s_k,a_k^i,a_k^j)\in \mathbb{R}$ is the reward received and $\os_{k+1}\in \oS$ is the next extended state. Agent $i$ approximates the expected continuation payoff by
\be\label{eq:sample}
v_k^i(\os_{k+1})\approx \sum_{s_+\in S}p(s_+\mid s_k,a^i_k,a^j_k) \cdot v^i_k(s_+,m_k+1),
\ee
as in the classical Q-learning algorithm \citep{ref:Sutton18}. The step sizes decaying to zero at a certain rate can address the impact of such stochastic approximation errors, e.g., see \citep{ref:Borkar08}.

Based on \eqref{eq:sample}, the agent can update $q_k^i(\os,a^i)$ for all $(\os,a^i)$ according to
\begin{align}
q_{k+1}^i(\os,a^i) = q_k^i(\os,a^i) + \overline{\alpha}_k(\os,a^i) \cdot \Big(\widehat{q}_k^i - q_k^i(\os,a^i)\Big).\label{eq:qupdate}
\end{align} 
The step size $\overline{\alpha}_k(\cdot)\in [0,1]$ is defined by 
\be\label{eq:stepsize}
\overline{\alpha}_k(\os,a^i) := \mathbb{I}_{\{(\os,a^i)=(\os,a_k^i)\}}\cdot \min\left\{1,\frac{\alpha_{c_k(\os_k)}}{\BR_k^i(a_k^i)}\right\}\quad\forall (\os,a^i),
\ee
where $\alpha_c\in [0,1]$ for $c=0,1,\ldots$ is some reference step size decaying to zero, i.e., $\alpha_c\rightarrow 0$ as $c\rightarrow\infty$, and $c_k(\os)\in \mathbb{N}$ denotes the number of times the extended state $\os\in\oS$ gets realized until stage $k$. The normalization in \eqref{eq:stepsize} ensures that the local Q-function estimates get updated synchronously for every action in the expectation conditioned on the history of the play \citep{ref:Leslie05}. Furthermore, the threshold from above by $1$, i.e., $\min\{1,\cdot\}$, ensures that the iterates remain bounded.

\begin{algorithm}[tb]
\caption{Episodic Individual Q-learning}\label{tab:algo}
\begin{algorithmic}[1]
\State {\bfseries initialize:} $q^{i}_0$, $v_0^i$ arbitrarily
\For{each stage $k=0,1,\ldots$}  
\State observe $s_k$ and construct $\os_k=(s_k,m_k)$
\If{$k>0$}
\State set $\widehat{q}_{k-1}^i = r_{k-1}^i + \gamma^i\cdot v_{k-1}^i(\os_{k})$
\State update $q_{k}^{i}(\cdot) = q_{k-1}^{i}(\cdot) + \overline{\alpha}_{k-1}(\cdot)\cdot(\widehat{q}_{k-1}^i - q_{k-1}^{i}(\cdot))$
\EndIf
\State play $a_k^{i}\sim \BR_k^{i}= \BR^i(q_k^i(\os_k,\cdot))$ and receive $r_k^i=r^i(s_k,a_k)$
\State update 
$v^{i}_k(\os_k) = \mathbb{I}_{\{m_k<M\}} \cdot q_k^{i}(\os_k,\BR^{i}_k)$ 
\State set $v_{k}^i(\os) = v_{k-1}^i(\os)$ for all $\os\neq \os_k$
\EndFor
\end{algorithmic}
\end{algorithm}

Algorithm \ref{tab:algo} provides an explicit description of the learning dynamic presented for agent $i$. Note that agent $i$ updates $q_{k}^{i}(\cdot)$ to $q_{k+1}^i(\cdot)$ at the beginning of stage $k+1$ after observing the next extended state $\os_{k+1}$ due to the one-stage look-ahead sampling. Algorithm \ref{tab:algo} also has a similar flavor with the decentralized Q-learning (Dec-Q) presented in \citep{ref:Sayin21} since both dynamics generalize the individual Q-learning to SGs. However, Algorithm \ref{tab:algo} differs from  the Dec-Q in the following ways: $(i)$ Algorithm \ref{tab:algo} uses a single step size, whereas the Dec-Q uses two step sizes decaying at different timescales. The single step size is relatively easier to implement for practical applications. 
$(ii)$ Algorithm \ref{tab:algo} has convergence guarantees beyond two-agent zero-sum SGs, as shown in Section \ref{sec:convergence}.  
$(iii)$ Algorithm \ref{tab:algo} is applicable for both discounted and time-averaged cases. 

\section{Convergence Results}\label{sec:convergence}

This section characterizes the convergence properties of Algorithm \ref{tab:algo} based on the following observation. Given the value function estimates $v_k^i(\cdot)$, the agents play an underlying stage game specific to the current extended state $\os_k$. The payoff functions of these stage games are given by the \textit{global Q-function} defined by\footnote{Agents could have tracked the global Q-function instead of the local one if they had access to opponent actions.}
\be\label{eq:Qk}
Q_k^i(\os_k, \cdot) := r^i(s_k, \cdot) + \gamma^i \sum_{s_+\in S} p(s_+\mid s_k,\cdot)\cdot v_k^i(s_+,m_k+1).
\ee
However, these payoff functions are not necessarily stationary due to their dependence on the value function estimates of the next extended states except the ones associated with the last substage, i.e., $m_k=M-1$, where we have $v_k^i(s_+,M-1) = 0$, and therefore,
\be\label{eq:lastQ}
Q_k^i(\os, a^i,a^j) := r^i(s, a^i,a^j)\quad\forall \os\in S\times \{M-1\}\mbox{ and }k.
\ee

Intuitively, the convergence of the dynamics can follow from backward induction since the individual Q-learning dynamics are known to converge equilibrium in two-agent zero-sum and identical-interest games played repeatedly \citep{ref:Leslie05}. However, there are two challenges: $(i)$ addressing the impacts of time-varying value function estimates and the sampling error \eqref{eq:sample} in the local Q-function estimates that get updated based on the local information only, and $(ii)$ ensuring that stage games have zero-sum or identical-interest structures. To address these challenges, the paper extends the convergence results for individual Q-learning dynamics to a sequence of two-agent games that become strategically equivalent to zero-sum or potential games in the limit. 

We say that a two-agent game $\langle A^i,A^j,u^i,u^j\rangle$ is a \textit{zero-sum game} if $u^i(a^i,a^j) + u^j(a^j,a^i) = 0$ for all $(a^i,a^j)\in A^i\times A^j$, and a \textit{potential game} if there exists some potential function $\Phi:\prod_{i=1}^2 A^i\rightarrow \mathbb{R}$ such that
\be\label{eq:potential}
u^i(a^i,a^j) - u^i(\tilde{a}^i,a^j) = \Phi(a^i,a^j) - \Phi(\tilde{a}^i,a^j)\quad\forall (\tilde{a}^i,a^i,a^j)
\ee
for each $i$. Furthermore, we say that two $n$-agent strategic-form games characterized by the tuples $\langle A^i,u^i \rangle_{i=1}^n$ and $\langle A^i,\underline{u}^i\rangle_{i=1}^n$ are \textit{strategically equivalent} with respect to the best response provided that there exists a positive constant $h^i>0$ and a function $g^i:A^{-i}\rightarrow\mathbb{R}$ such that \citep[Chapter 3]{ref:Basar99}
\be\label{eq:SE}
u^i(a^i,a^{-i}) = h^i \cdot \underline{u}^i(a^i,a^{-i}) + g^i(a^{-i})\quad\forall (a^i,a^{-i}).
\ee 

For explicit results, we consider a sequence of games 
\be
(\mathcal{G}_{k}= \langle \Omega,\mathbb{P}, (A^i,u_{k}^i)_{i=1}^N \rangle)_{k=0}^\infty,\label{eq:secGames}
\ee
where $\mathcal{G}_{k} $ involves Nature as a non-strategic player taking actions independently according to the probability distribution $\mathbb{P}$ over the action space $\Omega$. The payoffs $(u_k^i:\Omega\times A^i\times A^j\rightarrow\mathbb{R})_{k=0}^{\infty}$ form an exogenous sequence converging $u_k^i(\cdot)\rightarrow u^i(\cdot)$ pointwise as $k\rightarrow\infty$. Nature models the sampling error \eqref{eq:sample} and the convergent payoffs model the (exogenous) convergence of the value function estimates for the next stages. 

Algorithm \ref{tab:algo} reduces to the individual Q-learning for $|S|=1$ and $M=1$. Let us omit the dependence on the singleton state $s$ and sub-stage $m=0$ for notationally simplicity. Then, at each stage $k$, for the underlying game $\mathcal{G}_k$, agent $i$ follows the individual Q-learning dynamics given by
\begin{align}
&q_{k+1}^i(a^i) = q_k^i(a^i) + \overline{\alpha}_k(a^i)\left(u_k^i(\omega_k, a_k^i,a_k^j) - q_k^i(a^i)\right)\label{eq:indQ}
\end{align}
for all $a^i\in A^i$, where $\omega_k\sim \mathbb{P}$ is Nature's action. Note that we have $c_k(\os_k)=k$ in \eqref{eq:stepsize}.

We make the following assumption about $\{\alpha_k\}$.

\begin{assumption}\label{assume:stepsize}
The step sizes $(\alpha_k\geq 0)_{k=0}^{\infty}$ decay to zero at a certain rate such that 
$\sum_{k=0}^{\infty} \alpha_k = \infty$ and $\sum_{k=0}^{\infty} \alpha_k^2 < \infty$. 
\end{assumption}

The following theorem addresses the convergence of the individual Q-learning dynamics for the sequence of games that become strategically equivalent to zero-sum or potential games in the limit.

\begin{theorem}\label{thm:IndQ}
Consider the sequence of games, as described in \eqref{eq:secGames}.
Suppose that each agent $i$ follows the individual Q-learning dynamics \eqref{eq:indQ} for $\mathcal{G}_k$ at stage $k$ and Assumption \ref{assume:stepsize} holds. If the game $\mathcal{G} = \langle \Omega,\mathbb{P}, (A^i,u^i)_{i=1}^2\rangle$ is strategically equivalent to zero-sum or potential games, then we have $(q_k^1,q_k^2)\rightarrow (\overline{q}^1,\overline{q}^2)$ as $k\rightarrow\infty$ almost surely for some $\overline{q}^i:A^i\rightarrow\mathbb{R}$ satisfying 
\be
\overline{q}^i(a^i) =  u^i(a^i,\BR^j(\overline{q}^j))\quad\forall a^i \mbox{ and }i. 
\ee
\end{theorem}

\begin{proof} 
The proof can follow from the convergence analyses in \citep{ref:Leslie05} for the repeated play of two-agent zero-sum or potential games. In the following, we highlight the differences. 

Firstly, by turning the problem around, we can consider that agents play the game $\mathcal{G}$ repeatedly while they receive their payoffs $u^i(\omega_k,a_k^i,a_k^j)$ with some additive error $e_k^i := u_k^i(\omega_k,a_k^i,a_k^j) - u^i(\omega_k,a_k^i,a_k^j)$. The error $e_k^i$ is asymptotically negligible as $u_k^i(\cdot)\rightarrow u^i(\cdot)$ point-wise. Furthermore, the iterates remain bounded due to the threshold in the step size. The boundedness of the local Q-function estimates ensures that the threshold in \eqref{eq:stepsize} becomes redundant eventually as $\alpha_k\rightarrow 0$ since we can bound the probability $\BR_k^i(a_k^i)$ in the denominator of $\overline{\alpha}_k^i$ from below by some term $\varepsilon>0$ uniformly for all $k$. 

Based on Assumption \ref{assume:stepsize}, as in \citep[Lemma 3.1]{ref:Leslie05}, we can show that the iterates $q_k^i$'s converge almost surely to a connected internally chain-recurrent set of the flow, defined by the following ordinary differential equation (o.d.e.):
\begin{align}
&\frac{d\pq^i}{dt} (a^i) = u^i(\mathbb{P},a^i,\BR^j(\pq^j)) - \pq^i(a^i)\quad\forall a^i,\label{eq:flowq}
\end{align}
for all $i=1,2$ and $j\neq i$, where $\pq^i:[0,\infty)\rightarrow \mathbb{R}^{A^i}$ is the continuous-time counterpart of $q_k^i\in\mathbb{R}^{A^i}$. The o.d.e. \eqref{eq:flowq} and  \citep[Lemma 3.2]{ref:Leslie05} yield that $\mathcal{B}:=\{(u^1(\cdot,\pi^2),u^2(\cdot,\pi^1)): (\pi^1,\pi^2)\in \Delta^1\times\Delta^2\}$ is the global attractor of the flow \eqref{eq:flowq} and $q_k^i$'s are \textit{asymptotically belief-based} as 
$(\pq^1,\pq^2)\rightarrow \mathcal{B}$ for any initial $\pq^i(0)$'s. Furthermore, \citep[Lemma 4.1]{ref:Leslie05} yields that any connected internally chain-recurrent set of \eqref{eq:flowq} is included in the connected internally chain-recurrent set of the flow defined by the continuous-time best response dynamics
\be\label{eq:BRdyn}
\frac{d\ppi^i}{dt} = \BR^i(u^i(\mathbb{P},\cdot,\ppi^j)) - \ppi^i
\ee
for all $i=1,2$ and $j\neq i$. 

Let $\underline{\mathcal{G}}=\langle \Omega,\mathbb{P},(A^i,\uu^i)\rangle_{i=1}^2$ be either a zero-sum or potential game that is strategically equivalent to $\mathcal{G}$, as described in \eqref{eq:SE}. More explicitly, for each $i=1,2$ and $j\neq i$, there exists $h^i>0$ and $g^{i}:A^j\rightarrow\mathbb{R}$ such that
\be\label{eq:nSE}
u^i(\mathbb{P},a^i,a^j) = h^i \cdot \uu^i(\mathbb{P}, a^i,a^j) + g^i(a^j)\quad\forall (a^i,a^j).
\ee 
This implies that $\BR^i(u^i(\mathbb{P},\cdot,\ppi^j)) = \BR^i(\uu^i(\mathbb{P},\cdot,\ppi^j))$.
Therefore, the flow \eqref{eq:BRdyn} is identical to the continuous-time best response dynamics for the zero-sum/potential game $\underline{\mathcal{G}}$. Hence, we can complete the proof based on the Lyapunov functions formulated in \citep{ref:Hofbauer05} for zero-sum and potential games.
\end{proof}

Next, consider the directed graph $G=(S,E)$ with $S$ and $E\subset S\times S$ denoting, resp., set of vertices and edges. The vertices correspond to states and there is an edge from $s$ to $s_+$, i.e., $(s,s_+)\in E$, if there exists $a\in A$ such that $p(s_+\mid s,a) > 0$. 

\begin{assumption}\label{assume:state}
The graph $G$ satisfies the following conditions:
\begin{itemize}
\item[(a)] The graph $G$ is strongly connected, i.e., each vertex is reachable from every other vertex.
\item[(b)] There exists a vertex $s$ with cycle length $\ell$ such that $\gcd(\ell,M)=1$.
\end{itemize}
\end{assumption}

The condition (b) in Assumption \ref{assume:state} holds, e.g., $(i)$ if a vertex has self-loop, or $(ii)$ if $G$ is fully connected, i.e., the underlying SG is irreducible in the sense that for each $(s,s_+)$, there exists some $a$ such that $p(s_+\mid s,a)>0$, or $(iii)$ if $M$ is prime. 

Given the graph $G=(S,E)$, consider the graph $\overline{G} = (\oS,\overline{E})$ with $\oS=S\times \oM$ and $\overline{E}\subset \oS\times \oS$ denoting, resp., set of vertices and edges such that vertices correspond to extended states and there is an edge from $\os$ to $\os_+$, i.e., $(\os,\os_+)\in\overline{E}$, if there exists $a\in A$ such that $p(s_+\mid s,a) > 0$ and $m_+ \equiv (m+1) \mod(M)$.

\begin{proposition}\label{prop:io}
If the graph $G=(S,E)$ satisfies Assumption \ref{assume:state}, then the graph $\overline{G}=(\oS,\overline{E})$ is strongly connected.
\end{proposition}

\begin{proof}
Assumption \ref{assume:state}-(b) says that for some state $s$ we have $\gcd(\ell,M) = 1$. Then, B\'{e}zout's identity yields that there exist $x,y\in\mathbb{Z}$ such that $x\ell_{ss} + yM = 1$. Hence, for any $z\in \oM$, either $z\times x\ell \mod(M) \equiv z\mod(M)$ or $z\times x\ell \mod(M) \equiv M-z$. This implies that there is a path between the extended states $(s,m)$ and $(s,m')$ for any $(m,m')\in \oM\times \oM$. 

Consider arbitrary extended states $\os'=(s',m')$ and $\os''=(s'',m'')$. By Assumption \ref{assume:state}-(a), there exist paths from $s'$ to $s$ and $s$ to $s''$ in $G$ with lengths $\ell_{s's}$ and $\ell_{ss''}$, respectively. Then, as shown above, there exists a cycle for state $s$ in $G$ with length $\ell_{ss}$ such that $\ell_{s's} + \ell_{ss} + \ell_{ss''} \equiv (m''-m') \mod(M)$.
\end{proof}

 Proposition \ref{prop:io} yields that stage games associated with every extended state get played infinitely often, i.e., $c_k(\os)\rightarrow\infty$ as $k\rightarrow\infty$ almost surely for all $\os$, provided that the probabilities of any joint action gets played are uniformly bounded away from zero across $k=0,1,\ldots$. 
 
Algorithm \ref{tab:algo} may not converge to equilibrium in general as an uncoupled learning dynamic that does not incorporate the opponent's objective into the update rule \citep{ref:Hart03}. However, the following assumption identifies several widely-used classes of SGs for which we have provable convergence.

\begin{assumption}\label{assume:games}
The games $\langle A^i,r^i(s,\cdot)\rangle_{i=1}^2$ induced by rewards satisfy one of the following conditions:
\begin{itemize}
\item $\langle A^i,r^i(s,\cdot)\rangle_{i=1}^2$ is a zero-sum game for each $s$ and $\gamma^1=\gamma^2$ so that $\mathcal{S}$ is a zero-sum SG, or 
\item $\langle A^i,r^i(s,\cdot)\rangle_{i=1}^2$ is an identical-interest game for each $s$ and $\gamma^1=\gamma^2$ so that $\mathcal{S}$ is an identical-interest SG, or
\item $\langle A^i,r^i(s,\cdot)\rangle_{i=1}^2$ is either a zero-sum or potential game for each $s$ and there exists a single agent $i_s\in\{1,2\}$ such that 
\be
p(s_+\mid s,a_{i_s},(a_{j})_{j\neq i_s}) = p(s_+\mid s,a_{i_s},(\tilde{a}_{j})_{j\neq i_s})
\ee
for all $(s_+,s,a,(\tilde{a}_{j})_{j\neq i_s})$.
Agents can have different discount factors, including the time-averaged $(\gamma^1=1)$ vs discounted $(\gamma^2<1)$ cases.
\end{itemize}
\end{assumption}

The smoothed-fictitious-play variant studied in \citep{ref:Sayin22c} also have convergence guarantees for the SGs described in Assumption \ref{assume:games}. However, Algorithm \ref{tab:algo} differs from that algorithm by focusing on episodic strategies that reduce the coordination and memory burden on the agents. 

Now, we are ready to show the convergence of Algorithm \ref{tab:algo}.

\begin{theorem}\label{thm:SG}
Consider a two-agent SG $\mathcal{S}=\langle S,(A^i,r^i)_{i=1}^2,p\rangle$. Suppose that each agent follows Algorithm \ref{tab:algo} for some $M$, and Assumptions \ref{assume:stepsize}, \ref{assume:state}, and \ref{assume:games} hold. Then, we have $(q_k^i,v_k^i)_{i=1}^2 \rightarrow (q^i,v^i)_{i=1}^2$ as $k\rightarrow\infty$ almost surely for some $q^i:\oS\times A^i\rightarrow \mathbb{R}$ and $v^i:\oS\rightarrow\mathbb{R}$ satisfying
\begin{subequations}\label{eq:limit}
\begin{align}
&q^i(\os,a^i) = Q^i(\os,a^i,\opi^j(\os))\\
&v^i(\os) = \mathbb{I}_{\{m<M-1\}} \cdot q^i(\os,\opi^j(\os))
\end{align}
\end{subequations}
for all $(\os,a^i)$ and $i$, respectively, where the episodic strategy $\opi^j(\os) := \BR^j(q^j(\os,\cdot))$ and
\be\label{eq:QQ}
Q^i(\os,a) := r^i(s,a)+\gamma^i\sum_{s_+\in S} p(s_+\mid s,a) \cdot v^i(s_+,m+1).
\ee

Furthermore, we have $U^i(\opi^i,\opi^j;s) \geq U^i(\pi^i,\opi^j;s) - \epsilon^i$ for all $\pi^i\in\Pi^i$ and $(i,j)\in\{(1,2),(2,1)\}$. The approximation error $\epsilon^i\geq 0$ is given by
\be\label{eq:bound2}
\epsilon^i := \left\{\begin{array}{ll}
\delta^i\cdot \cfrac{1}{M} + \tau \log|A^i| & \mbox{if } \gamma^i=1\\
\delta^i\cdot \cfrac{(1-\gamma^i)(\gamma^i)^M}{1-(\gamma^i)^M} + \tau \log|A^i| &\mbox{if } \gamma^i\in [0,1)
\end{array}\right., 
\ee
where $\delta^i\geq 0$ is as described in \eqref{eq:deltai}.
\end{theorem}

\begin{proof}
Consider the sequence of stage games
$\left(\mathcal{G}_{k}(\os)=\langle \Omega,\mathbb{P},(A^i,\widehat{Q}_k^i(\os,\cdot))_{i=1}^2\rangle\right)_{k=1}^{\infty}$,
specific to each extended state $\os$. The payoffs $\widehat{Q}_k^i(\os,\cdot):\Omega\times A^i\times A^j\rightarrow\mathbb{R}$ are defined by
\be
\widehat{Q}_k^i(\os,\omega,a) := r^i(s,a) + \gamma^i \cdot v_k^i(s'(\omega,s,a),m+1),
\ee
where the function $s':\Omega\times S\times A\rightarrow S$ is defined such that $s'(\omega,s,a)=s_+$ with probability $p(s_+\mid s,a)$ (due to the randomness on $\omega\sim \mathbb{P}$) so that we have $\widehat{Q}_k^i(\os,\mathbb{P},\cdot)\equiv Q_k^i(\os,\cdot)$ for $Q_k^i$, as described in \eqref{eq:Qk}. This implies that at stage $k$, the agents play the stage game $\mathcal{G}_k(\os_k)$, associated with the current extended state $\os_k\in\oS$.

On the other hand, the threshold on the step sizes used in Algorithm \ref{tab:algo} ensures that the iterates remain bounded. Given any $q^i:A^i\rightarrow \mathbb{R}$ bounded by $\max_{a^i}|q^i(a^i)| \leq \bar{q}^i$, the smoothed best response \eqref{eq:smoothed} satisfies
\be\label{eq:lowerbound}
\BR^i(q^i)(a^i) = \frac{\exp(q^i(a^i)/\tau)}{\sum_{\tilde{a}^i\in A^i}\exp(q(\tilde{a}^i/\tau))}  \geq \frac{1}{|A^i|}\cdot \exp(2\bar{q}/\tau)\quad\forall a^i.
\ee
Due to \eqref{eq:lowerbound} and Assumption \ref{assume:state}, Proposition \ref{prop:io} yields that every stage game associated with each extended state get played infinitely often. 

By \eqref{eq:lastQ}, the stage games for the last substage (i.e., $m=M-1$) of the episodes are stationary. As these stage games get played infinitely often and Assumption \ref{assume:stepsize} holds, Theorem \ref{thm:IndQ} yields that  $q_k^i(s,M-1,a^i)\rightarrow q^i(s,M-1,a^i)$ and correspondingly $v_k^i(s,M-1)\rightarrow v^i(s,M-1)$ for all $(s,a^i)$ and $i$ almost surely, where $q^i$ and $v^i$ are as described in \eqref{eq:limit}. This implies that 
\be
\mathcal{G}_k(s,M-2)\rightarrow \langle \Omega,\mathbb{P},(A^i,Q^i(s,M-2,\cdot))_{i=1}^2\rangle
\ee 
for each $s$, almost surely. The backward induction implies the convergence of $(q_k^i(\os,\cdot),v_k^i(\os))$ for all $\os\in \oS$. 

Next, we evaluate the performance of $\opi^i$ against $\opi^{j}$ for the finite-horizon utility $\wU^i(\cdot)$, as in \eqref{eq:nearfinite}, to quantify the approximation level of the episodic strategy profile $(\opi^i)_{i=1}^2$ based on Theorem \ref{lem:neari}. To this end, let $\opi_*^i$ be an $M$-episodic strategy that is the best response to $\opi^{j}$ for $\wU^i(\cdot)\equiv U_{0:M-1}^i(\cdot)$, where $U_{\underline{m}:\overline{m}}^i(\cdot)$ is as described in \eqref{eq:uu}. Then, the principle of optimality yields that
\begin{align}\label{eq:prencip}
&U_{m:M-1}^i(\opi_*^i,\opi^{j};s) = \max_{a^i\in A^i} \bigg\{r^i(s,a^i,\opi^{j}(s,m))\nonumber \\
&\hspace{.2in}
+ \gamma^i \sum_{s_+\in S} p(s_+\mid s,a^i,\opi^{j}(s,m)) \cdot U_{m+1:M-1}^i(\opi_*^i,\opi^{j};s_+)\bigg\}
\end{align} 
for each $s$ and $m\in\oM$. On the other hand, by \eqref{eq:smoothed}, the smoothed best response $\BR^i(q^i)$ achieves
\begin{equation}\label{eq:epsilon}
0\leq  \max_{a^i\in A^i} \{q^i(a^i)\} - q^i(\BR^i(q^i)) \leq \tau\log|A^i| =:\xi^i
\end{equation}
for any $q^i:A^i\rightarrow\mathbb{R}$. Define the sequence $\{e_m^i\}_{m=0}^{M-1}$ recursively such that $e_{M-1} = \xi^i$ and 
$e_m^i := \xi^i + \gamma^i \cdot e_{m+1}^i$ for $m<M-1$. 
Then, by \eqref{eq:epsilon}, at the last substage, i.e., $m=M-1$, we have 
\begin{equation}
U_{M-1:M-1}^i(\opi_*^i,\opi^{j};s) \leq v^i(s,M-1) + e_{M-1}^i\quad\forall s.
\end{equation} 
For some $m < M-1$, suppose that 
$U_{m+1:M-1}^i(\opi_*^i,\opi^{j};s) \leq v^i(s,m+1) + e_{m+1}^i$ for all $s$. Then, by \eqref{eq:QQ}, \eqref{eq:prencip}, and \eqref{eq:epsilon}, we have
\begin{align}
U_{m:M-1}(\opi_*^i,\opi^{j};s) &\leq \max_{a^i\in A^i} \bigg\{r^i(s,a^i,\opi^{j}(s,m))\nonumber \\
&\hspace{-.4in}
+ \gamma^i \sum_{s_+\in S} p(s_+\mid s,a^i,\opi^j(s,m)) \cdot (v^i(s_+,m+1) + e_{m+1}^i)\bigg\}\nonumber\\
&\leq v^i(s,m-1) + \xi^i + \gamma^i\cdot e_{m}^i,
\end{align} 
Therefore, we obtain
$U_{m:M-1}^i(\opi_*^i,\opi^{j};s) \leq v^i(s,m) + e_m^i$. By induction, we can conclude that $(\opi^i)_{i=1}^2$ satisfies \eqref{eq:nearfinite} for $\weps^i = e_0^i$ and $e_0^i=\xi^i\sum_{k=0}^{M-1}(\gamma^i)^k$ by its definition. Incorporating $\weps^i = \xi^i\sum_{k=0}^{M-1}(\gamma^i)^k$ into \eqref{eq:bound}, we can obtain \eqref{eq:bound2}. 
\end{proof}

\section{Conclusion}\label{sec:conclusion}

This paper presented finite-horizon approximation scheme and episodic equilibrium as a solution concept for SGs. The paper established approximation guarantees for this scheme, demonstrating its effectiveness in bridging the gap in the analyses of finite-horizon and infinite-horizon SGs. Furthermore, the paper introduced episodic individual Q-learning dynamics that provably converge to (near) episodic equilibrium in various important SG classes, encompassing both time-averaged and discounted utilities. Some of the future research directions include the study of 
episodic variants of other equilibrium concepts such as correlated equilibrium and Stackelberg equilibrium, and
episodic variants of other learning dynamics or equilibrium computation methods or equilibrium-seeking algorithms.

\section*{Acknowledgment}
This work was supported by The Scientific and Technological Research Council of T\"{u}rkiye (TUBITAK) BIDEB 2232-B International Fellowship for Early Stage Researchers under Grant Number 121C124.

\end{spacing}

\begin{spacing}{1}
\bibliographystyle{plainnat}
\bibliography{mybib}
\end{spacing}

\end{document}